\documentclass[12pt,reqno]{article}
\usepackage{fullpage}
\usepackage{epigraph}
\usepackage{amscd,amsmath,amsthm}
\usepackage{xspace,relsize}
\newcommand{\seqnum}[1]{\cite[\href{http://oeis.org/#1}{\underline{#1}}]{OEIS}}
\let\seqnump\seqnum
\usepackage{graphicx,amssymb,mathtools}
\usepackage{pstricks}
\usepackage{pst-node,pst-tree}
\usepackage{epstopdf}
\newcommand{\darkcirc}{\bullet}
\newcommand{\Bx}{\mbox{$\lozenge$}}
\newcommand{\Lf}{\TR{$\Bx$}}
\newcommand{\La}{\left\langle}
\newcommand{\Ra}{\right\rangle}
\newcommand{\Lt}{\mbox{\scriptsize$\triangle$}}
\newcommand{\Dt}{\mbox{$\blacktriangle$}}
\newcommand{\Lc}{\mbox{$\circ\!\!\circ\!\!\circ$}}
\newcommand{\Dc}{\mbox{$\darkcirc\!\!\darkcirc\!\!\darkcirc$}}
\newcommand{\NE}{\ensuremath\nearrow}
\newcommand{\SE}{\ensuremath\searrow}
\newtheorem{theorem}{Theorem}
\newtheorem{corollary}[theorem]{Corollary}
\newtheorem{proposition}[theorem]{Proposition}
\title{Nonleaf Patterns in Trees:\\ Protected Nodes and Fine Numbers}
\author{Nachum Dershowitz\\
School of Computer Science, Tel Aviv University\\Ramat Aviv, Israel\\
\href{mailto:nachum@cs.tau.ac.il}{\tt nachum@cs.tau.ac.il}}
\date{\today}
\begin{document}
\maketitle

\epigraph{\small The barren and fertile fronds [of the Sensitive Fern] are extremely unlike, the former being leaf-like, very sensitive to frost, quickly wilting when plucked, and much
taller and more common than the latter which are non-leaf-like and remain erect, though drying up, through the winter.}{---\textit{Transactions of the Royal Society of Canada} (1884)}

\begin{abstract}
A closed-form formula is derived for  the number of occurrences of matches of a multiset of patterns
among all  ordered (plane-planted) trees with a given number of edges.
A pattern looks like a tree, with internal nodes and leaves,
but also contain components that match subtrees or sequences of subtrees.
This result extends previous versatile tree-pattern enumeration formul\ae{} to incorporate components that are only allowed to match nonleaf subtrees 
and provides enumerations of trees by the number of protected
(shortest outgoing path has two or  more edges) or unprotected  nodes.
\end{abstract}

\section{Tree Patterns}

One routinely asks how often certain patterns occur among a set of combinatorial objects, such as trees.
Here, we concern ourselves with patterns in ordered (plane-planted) trees, which
are counted by the Catalan numbers, $C_n=\frac1{n+1}\binom{2n}{n}$,
sequence \seqnum{A000108} in Sloane's \textit{Encyclopedia of Integer Sequences}~\cite{OEIS}. 
Lattice (Dyck) paths---which may not dip below the abscissa, and very many other interesting objects (including full binary trees), are their equinumerous analogues; see \cite[e.g.~Thm.~1.5.1]{Stanley}.
Ordered trees may be written down linearly as balanced nested brackets.
Each balanced bracket pair $\La\cdots\Ra$ corresponds to a single node.
For example, $\La\,\La\La\Ra\!\La\Ra\Ra\La\Ra\La\La\La\Ra\Ra\Ra\,\Ra$ is a seven-edge, eight-node tree, with four leaves $\La\Ra$
and four (nonleaf) internal nodes. 
Its root has three children, the middle one being a (childless) leaf. 
We refer to such a level-one leaf as a \emph{stump}.
This tree is drawn in Figure~\ref{fig}, with leaves labeled \textsf{w}, \textsf{x}, \textsf{y}, and \textsf{z}; it
corresponds to the path depicted in Figure~\ref{fig:path}, with peaks at $x=2,4,7,11$.
The tree stump \textsf{y} turns into the little hill in the middle of the path at $x=7$.
Stump-less trees are deemed ``protected'' and form the subject of Section~\ref{sec:pro}.
Hills  and stumps are further explored in Section~\ref{sec:stump}.

\begin{figure}[t]
\centering
\psset{levelsep=10mm,nodesep=-1pt,radius=1.5mm,tnpos=r,tnsep=1pt,tnyref=0.4}
\begin{pspicture}(-4,-3)(5,1){
\pstree{\TC~{\textsf{f}}}{
\pstree{\TC~{\textsf{g}}}{{\Lf~{\textsf{w}}}{\Lf~{\textsf{x}}}}
{\Lf~{\textsf{y}}}
\pstree{\TC~{\textsf{h}}}{\pstree{\TC~{\textsf{k}}}{{\Lf~{\textsf{z}}}}}}
}
\end{pspicture}
\caption{The 7-edge 4-leaf tree $\La\,\La\La\Ra\!\La\Ra\Ra\La\Ra\La\La\La\Ra\Ra\Ra\,\Ra$ with labeled nodes for reference.}\label{fig}
\end{figure}
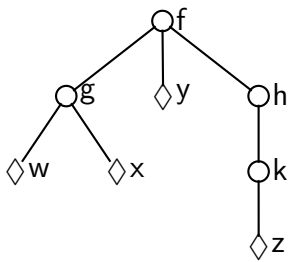

\begin{figure}[t]
\centering
\begin{pspicture}(0,0)(14,4.5)
\rput(6.7,3){\includegraphics[bb=120 150 620 700,scale=.09]{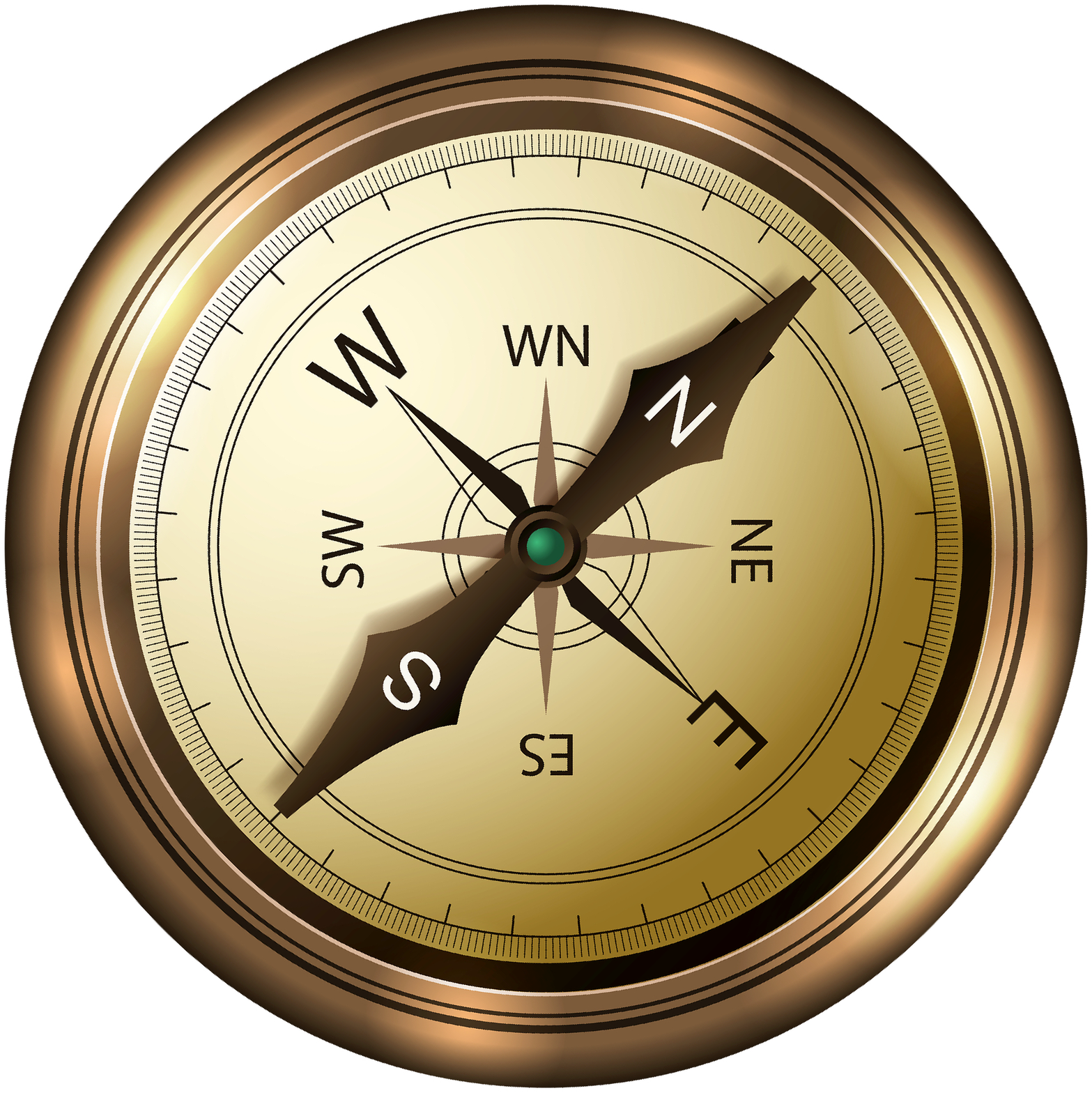}}
\psgrid[subgriddiv=1,griddots=14](0,0)(0,0)(14,4)
\psset{linecolor=red}
\pcline[linestyle=solid]{-}(0,0)(14,0)
\pcline[linestyle=solid]{-}(0,4)(14,4)
\psset{linecolor=blue,arrowsize=6pt}
\pcline[linestyle=solid]{->}(0,0)(1,1)
\pcline[linestyle=solid]{->}(1,1)(2,2)
\pcline[linestyle=solid]{->}(2,2)(3,1)
\pcline[linestyle=solid]{->}(3,1)(4,2)
\pcline[linestyle=solid]{->}(4,2)(5,1)
\pcline[linestyle=solid]{->}(5,1)(6,0)
\pcline[linestyle=solid]{->}(6,0)(7,1)
\pcline[linestyle=solid]{->}(7,1)(8,0)
\pcline[linestyle=solid]{->}(8,0)(9,1)
\pcline[linestyle=solid]{->}(9,1)(10,2)
\pcline[linestyle=solid]{->}(10,2)(11,3)
\pcline[linestyle=solid]{->}(11,3)(12,2)
\pcline[linestyle=solid]{->}(12,2)(13,1)
\pcline[linestyle=solid]{->}(13,1)(14,0)
\end{pspicture}
\medskip
\caption{The lattice path corresponding to the tree in Figure~\ref{fig}, consisting of 7 \NE-steps (pointing to the North) and 7 \SE-steps (East).}\label{fig:path}
\end{figure}

Motivated by enumeration problems such as those in \cite{Fine,protected}, we 
add the possibility of a sequence of nonleaf children to the versatile pattern enumerations
of~\cite{Pat,More}.
With this addition, patterns come in five basic shapes:
$\Bx$, $\Lt$, $\Dt$, $\Lc$, and $\Dc$.
\begin{itemize}\def\labelitemi{--}
\item A lozenge $\Bx$ corresponds to any tree leaf.
\item A light triangle pattern $\Lt$ matches any subtree.
\item A dark triangle $\Dt$ matches any \emph{nonleaf} subtree (that is,
a subtree rooted at an internal node).
\item A light ellipsis $\Lc$ can match any sequence of (zero or more) subtrees.
\item A dark ellipsis $\Dc$ can match any sequence of (zero or more) \emph{nonleaf} subtrees.
\end{itemize}

Basic patterns can be \emph{composed} to form more complicated shapes
by placing a sequence of patterns in angle brackets.
Specifically, tree patterns $P$ have the following grammar:
\begin{align*}
%P &::= \Bx \mid \Lt \mid \Dt \mid  \Lc \mid \Dc \mid \La P\cdots P\Ra 
P &::= \Bx \mid \La Q\, Q\cdots Q\Ra\\
Q &::= \Bx \mid \Lt \mid \Dt \mid  \Lc \mid \Dc \mid \La Q\cdots Q\Ra 
\end{align*}
%where $P\cdots P$ means one or more patterns $P$, in sequence.
where $Q\cdots Q$ means zero or more patterns $Q$, in sequence,
and $Q\,Q\cdots Q$ means one or more.
The lozenge pattern $\Bx$ is the same as $\La\Ra$,
but we will need to control where leaves appear, so we only allow $\La\Ra$ 
within composites and require the use of $\Bx$ as a top-level pattern.
The triangle pattern $\Lt$ matches any subtree matched by either $\Bx$ or $\Dt$;
the ellipsis  $\Lc$ is tantamount to $\Lt\!\cdots\!\Lt$,
and the dark ellipsis  $\Dc$ is like  $\Dt\!\cdots\!\Dt$.
Both $\La\Lc\Ra$ and $\La\Dc\Ra$ match leaves, as they have zero subtrees.
(The basic patterns of~\cite{More} did not include the dark ellipsis;
those of~\cite{Pat} also excluded dark triangles.)
Given a bag (multiset) of patterns, we ask how many times those patterns occur in trees of a specified size.
Patterns may occur more than once or not at all in any given tree.

Two patterns \emph{overlap} if their instances share one or more nodes in the tree.
We aspire to count non-overlapping instances only.
More than one triangle or ellipsis can co-occur at the same point in a tree, without consuming any nodes.
For example, \Lt\@ and \Dt\@ each match any nonleaf subtree, 
but have no nodes of their own;
were we to want to count their joint occurrences in a tree,
we would not want the  two of them to match the very same node at the same  time.
As  we are interested in counting distinct, non-overlapping occurrences of patterns, triangles and ellipses are not used as standalone patterns ($P$), but rather only within composites ($Q$).

The leaf pattern $\Bx$ matches every childless (leaf) node, of which there are four
in the example tree of Figure~\ref{fig}.   
The pattern $\La\Lt\Ra$ matches ``only children'', of which there are two in the figure (\textsf{h}  and \textsf{k}).
The pattern $\La\Lc\,\Bx\Ra$ matches each node whose youngest (rightmost) child is childless;
that also happens twice in the example (at \textsf{g}  and at \textsf{k}).
So, a \emph{pair} of two copies of this pattern shows up exactly once (at the node pair \textsf{g} \& \textsf{k}), covering those two occurrences;
three such, not at all.
The pattern $\La\Dc\Bx\Ra$ matches a node whose youngest child is the only one who is not a parent;
there is one such (\textsf{k}).
The pattern $\La\Lt\Lc\Bx\Lc\Ra$ matches a leaf provided it is not the eldest child;
there are two (\textsf{x}  and \textsf{y}).
The pattern $\La\Lt\Lc\Dt\Lc\Lt\Ra$ would match a nonleaf middle child  were there one.
The pair of patterns  $\La\Bx\Lc\Ra$ and  $\La\Lc\Bx\Ra$, having a leftmost leaf and having a rightmost leaf, occurs twice in the tree: the first pattern at node \textsf{g} on account of its leaf-child \textsf{w} and the second at \textsf{k} with  its childless child \textsf{z}; and reversed, at \textsf{k}-\textsf{z}  and \textsf{g}-\textsf{x}.
Though these two patterns do co-occur at \textsf{g} and also at \textsf{k}, such overlapping occurrences do not count. 

\section{Pattern Enumerations}

We need to know the number of tree nodes that appear within patterns,
which we calculate as follows:
\begin{align*}
v(\Lt) &= 
v(\Dt) = 0 &
v(\Lc) &= 
v(\Dc) = 0\\
%e(\Lc) &= 0\\
v(\Bx) &= 1&
%e(\Dc) &= 0\\
v(\La p_1\cdots p_k\Ra) &= 1+{\mathlarger\Sigma}_j v(p_j)
\,.
%e(\La p_1,\dots,p_a\Ra) &= 1+\Sigma_j e(p_j)
\end{align*}

Our primary enumeration result is the following:

\begin{theorem}[Pattern Enumeration]\label{thm:new}
Let $p_1,\ldots,p_q$, $q\geq 1$, be various composite patterns.
The number of non-overlapping occurrences among all $n$-edge ordered trees
of $n_j$ of each of the patterns $p_j$ and of $\ell\geq 0$ leaf patterns $\Bx$
is
\begin{align}
\binom{m}{n_1,\ldots,n_q}
\sum\limits_{k=1}
\frac{1}{k}
\binom{k}{m}
\binom{u-k}{\ell}
\sum\limits_{i=0}^{m+e-k}
\binom{d+i-1}{i}
\binom{e+b-i-1}{b-m+k-1}
\binom{e+a-i}{u-k}
\,, %\tag{$\ast$}
\label{main}
\end{align}
where
\begin{itemize}
\item $a$ is the number of light triangles $\Lt$ appearing in them,
\item $b$ is the number of light ellipses $\Lc$, 
\item $c$ is the number of dark triangles $\Dt$,
\item $d$ is the number of dark ellipses $\Dc$, 
\item $m=\Sigma n_j$ is the total number of patterns---presumed to be nonzero,
\item $v=\Sigma_j v(p_j)$ is the total number of nodes in the composite patterns, 
\item $e=n+m-v-c-a$ is the number of edges not accounted for by the patterns,
and
\item $u=n+m-v+1$.
\end{itemize}
\end{theorem}
%We ascribe to the convention that $\binom x0=1$ for all $x$, positive or negative.

\begin{proof}
The $m$ nonleaf patterns leave $n+1-v=u-m$ nodes unaccounted for,
any of which could be a leaf.
The patterns require at least $\ell$ of them to be leaves.
We count separately for
each possible number of ``loose''
(unattached to composite patterns) tree leaves, $k=\ell,\ell+1,\ldots,u-m$.
The number of tree nodes that are accounted for by the original patterns and all these leaves
is $v+k$.

The proof proceeds in several steps:
\begin{enumerate}

\item Arrange the given $m$ nonleaf patterns in a row, in any of
$\binom{m}{n_1,\ldots,n_q}$ ways.

\item\label{2} Intersperse an extra $n+1-v-k=u-m-k$ patterns of the form $\La\Lt\Lc\Ra$
among the $m$ patterns,
to cover all the missing internal (nonleaf) nodes,
in ${\binom{u-k}{m}}$ ways, for a total of $u-k$ patterns.

\item There are $e-(n+1-v-k)=m+e-u+k$ missing edges (of the $e$ missing from
the given patterns; $n+1-v-k$ were just added in the previous step).
Split them into two categories: $i$ edges that may not take leaves and $v+k-n+e-i-1$ that may.
This adds a summation $\Sigma_i$.

\item
Distribute these $i$ edges as sequences $\Dt\cdots\Dt$ of restricted triangles,
in place of the $d$ restricted ellipses $\Dc$.
There are
 ${\binom{d+i-1}{i}}$
ways to do this.

\item\label{5}
Similarly,
distribute the remaining $v+k-n+e-i-1$ edges as sequences $\Lt\cdots\Lt$ of unrestricted triangles,
in place of the $b+n+1-v-k$ light ellipses $\Lc$ ($b$ original and $n+1-v-k$ added).
There are
 ${\binom{e+b-i-1}{n-v+b-k}}$
ways to do this.

\item Place the $k$ leaves in some of the
$e+a-i$ unrestricted, light triangles
($a$ in the original patterns, plus
$n+1-v-k$ from step \ref{2} and $v+k-n+e-i-1$ from step \ref{5}),
in ${\binom{e+a-i}{k}}$ ways.

\item Select $\ell$ leaves to match those in the given patterns in $\binom{k}{\ell}$ ways.

\item The cyclic arrangement of the resultant $m+(u-m-k)=u-k$ 
patterns corresponds to exactly one occurrence of the patterns in a tree.
(This is an application of Dvoretsky and Motzkin's Cycle Lemma~\cite{DM}; see~\cite{CL}.)
To see this, graft the patterns into one tree by repeatedly
picking any pattern in the sequence and
inserting it into the closest (rightmost) available triangle slot among the patterns
preceding it, wrapping back around from the end when necessary.
The $u-k$ patterns contain a total of 
$a+c+(u-m-k)+(m-a-c+k-1)-\ell-(k-\ell)=u-k-1$ slots.
So, in fact, a single tree results from the grafting,
with each pattern occurring at the point it ends up in the reconstructed tree.
Thus, the enumeration has an additional factor of $\frac{1}{u-k}$.
\end{enumerate}

Collecting everything and summing for $k$, we have
\begin{align}\label{alt}
\binom{m}{n_1,\ldots,n_q}
\sum\limits_{k=\ell}^{u-m}
&\frac{1}{u-k}{\binom{u-k}{m}}\binom{k}{\ell}
\sum\limits_{i=0}^{m+e-u+k}\!
{\binom{d+i-1}{i}}{\binom{e+b-i-1}{n-v+b-k}}{\binom{e+a-i}{k}}\,.
\end{align}

Reversing the order of summation for $k$ (swapping $k$ and $u-k$) and
avoiding a 0 denominator, 
gives the stated enumeration (\ref{main}).
The sum for $k$ in (\ref{main}) can run from $\max\{1,m\}$ to $u-\ell$.
\end{proof}

For instance, nodes parenting exactly 2 (childless) leaves 
match the pattern $\La \Dc \Bx \Dc \Bx\Dc \Ra$.
So the number of such among the 14 size-4 trees (having 4 edges and 5 nodes) is obtained by setting $n=4$, $m=1$, $\ell=a=c=b=0$, $d=v=3$,  $e=2$, and $u=3$:
\begin{align*}
&
\binom{1}{1}
\sum\limits_{k=1}
\frac{1}{k}\binom{k}{1}
\binom{3-k}{0}
\sum\limits_{i}^{}\!
\binom{2+i}{i}
\binom{1-i}{k-2}
\binom{2-i}{3-k}\\
{}={}&
\sum\limits_{i=0}^1
\binom{2+i}{i}
\sum\limits_{k=2}^3
\binom{1-i}{k-2}
\binom{2-i}{3-k}
= [2+1]+3[1]=6\,.
\end{align*}
See Figure~\ref{fig:pro}, where the pattern matches one node in each of the second through sixth trees and one in the penultimate tree.

Whenever there can be at most one occurrence of the patterns per tree,
our formula counts trees---rather than mere instances of patterns.
This is the case, in particular, when the patterns cover each of the $n+1$ nodes and when those patterns are \emph{unambiguous}, in the sense that only one of the patterns
can match at any one of the nodes.

Call any internal node \emph{protected}~\cite{protected} when it is a grandparent via each and every one of its children,
and \emph{unprotected} when at least one child is childless.
This notion of protection was recently extended to $k$-protection, that is, that no path from the node contains fewer than $k$ edges,  in \cite{kpro}.  For the time  being, $k=2$.

The unambiguous pattern $\La\Dc\Bx\Lc\Ra$ matches each unprotected internal node 
in a unique manner, with the leftmost leaf child singled out.
Another way of looking at this pattern is as counting the eldest among leaf siblings.
There are 24 such in Figure~\ref{fig:pro}; the other 11 leaves in the figure have childless elder siblings.
On the other hand, $\La\Lc\Bx\Lc\Ra$, though it also matches unprotected nodes, it does so as many times as a node has childless children (viz.\@ 35 times in Figure~\ref{fig:pro}).
So it counts leaf children, rather than counting nodes having  leaf children.

\begin{figure}[t]
\centering
\psset{treefit=loose,levelsep=5mm,treesep=2mm,radius=1.0mm,nodesep=-1pt,dotsize=2mm,treenodesize=2mm}
\begin{pspicture}(-3.,-1.5)(5,0.5)
\pstree{\TC}{\Lf\Lf\Lf\Lf}
\pstree{\TC}{\pstree{\TC}{\Lf}\Lf\Lf}
\pstree{\TC}{\Lf\pstree{\TC}{\Lf}\Lf}
\pstree{\TC}{\Lf\Lf\pstree{\TC}{\Lf}}
\pstree{\TC}{\pstree{\TC}{\Lf\Lf}\Lf}
\pstree{\TC}{\Lf\pstree{\TC}{\Lf\Lf}}
\end{pspicture}\\
\psset{treesep=4mm}
\begin{pspicture}(-3.,-2)(5,0.5)
\pstree{\Tdot}{\pstree{\TC}{\Lf}\pstree{\TC}{\Lf}}
\pstree{\TC}{\pstree{\Tdot}{\pstree{\TC}{\Lf}}\Lf}
\pstree{\TC}{\Lf\pstree{\Tdot}{\pstree{\TC}{\Lf}}}
\pstree{\Tdot}{\pstree{\TC}{\Lf\Lf\Lf}}
\pstree{\Tdot}{\pstree{\TC}{\pstree{\TC}{\Lf}\Lf}}
\pstree{\Tdot}{\pstree{\TC}{\Lf\pstree{\TC}{\Lf}}}
\pstree{\Tdot}{\pstree{\Tdot}{\pstree{\TC}{\Lf\Lf}}}
\pstree{\Tdot}{\pstree{\Tdot}{\pstree{\Tdot}{\pstree{\TC}{\Lf}}}}
\end{pspicture}
\caption{The $C_4=14$ four-edge trees with 11 protected nodes in black.}\label{fig:pro}
\end{figure}
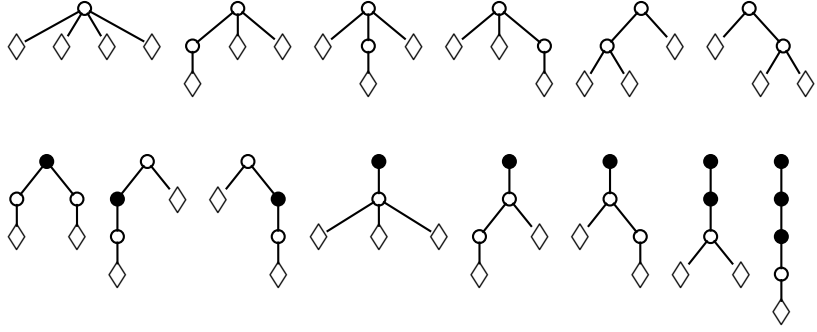

Whenever there are no dark ellipses ($d=0$) in the patterns,
$i$ only takes 0 in formula (\ref{alt}), 
for otherwise $\binom{d+i-1}{i}=0$.
For this to be possible, we also need for
$k\geq u-m-e$, or else the sum over $i$ is empty.
Accordingly, formula (\ref{alt}) simplifies substantially:
\begin{align}\nonumber
&
\binom{m}{n_1,\ldots,n_q}
\sum\limits_{k=u-m-e}^{u-m}
\frac{1}{u-k}
\binom{u-k}{m}
\binom{k}{\ell}
\binom{e+b-1}{n-v+b-k}
\binom{e+a}{k}
\\
{}={} &
\binom{m}{n_1,\ldots,n_q}
\binom{e+a}{\ell}
\sum\limits_{k=u-m-e}^{u-m}
\frac{1}{u-k}
\binom{u-k}{m}
\binom{e+a-\ell}{k-\ell}
\binom{e+b-1}{n-v+b-k}
\,,\label{dz}
\end{align}
with $e=n+m-v-c-a$ and  $u=n+m-v+1$.
This matches the main result of \cite[Thm.~4.1]{More} with various notational changes.
When, in addition, there are no dark triangles ($c=0$), this reduces
to a much simpler formula, as given in \cite[Thm.~2.1]{Pat}, namely,
\begin{align*}
\frac{1}{n-v-\ell+1}\binom{n+m-v}{n_1,\ldots,n_q,\ell,n-v-\ell}\binom{2n+m-2v-\ell-a+b}{n-v-\ell+b}
\,.
\end{align*}

The statement of the main theorem above assumed that $m>0$,
but the formul\ae\@ still make sense for $m=0$, as long as we take care to avoid a 0 denominator.
When there are no composite patterns and, hence, $m=0$, 
only the number of leaf patterns $\ell\geq 0$ is given.
So, all that is being counted is the number of occurrences of $\ell$ leaves within trees comprising $n$ edges.
A tree with $k\geq\ell$ leaves has $\binom{k}{\ell}$ such instances.
With no composite patterns, all of $a,b,c,d,v$ must be zero, whereas $e=n$ and $u=n+1$.
So formula (\ref{dz}) becomes 
\begin{align}\nonumber
\binom{n}{\ell}
\sum\limits_{k=\ell}^{n}
\frac{1}{n+1-k}
\binom{n-\ell}{k-\ell}
\binom{n-1}{n-k}
&=
\frac{1}{n}
\binom{n}{\ell}
\sum\limits_{k}^{}
\binom{n-\ell}{k}
\binom{n}{n-\ell+1-k}
\\\label{leaves}
&=
\frac{1}{n}
\binom{n}{\ell}
\binom{2n-\ell}{n-1}
\,.
\end{align}
For example, for $n=4$ and $\ell=3$, there are 6 trees with exactly three leaves
plus one tree with four leaves, leaving 4 ways of choosing just three of them, for a total of
$
\frac{1}{4}
\binom{4}{3}
\binom{8-3}{4-1}
=
10
$
occurrences of a triad of leaves.
See Figure~\ref{fig:pro}.

At the opposite end of the spectrum,
when 
all edges are accounted for by the patterns, that is, 
when $e=0$, there is nothing left to place in any ellipsis,
so we can ignore them for all intents and purposes,
imagining $b=d=0$.
We have $u=a+c+1$ and $v=n+m-c-a$.
Looking at (\ref{dz}), we obtain
\begin{align}\nonumber
&
\binom{m}{n_1,\ldots,n_q}
\binom{a}{\ell}
\sum\limits_{k=a+c+1-m}^{a+c+1-m}
\hspace{-10pt}\frac{1}{a+c+1-k}
\binom{a+c+1-k}{m}
\binom{a-\ell}{k-\ell}
\binom{-1}{c+a-m-k}
\,,
\end{align}
which simplifies to
\begin{align}
\frac{1}{m}
\binom{m}{n_1,\ldots,n_q}
\binom{a}{\ell}
\binom{a-\ell}{m-c-1}
\,.
\end{align}

For example, (full) binary trees with $2r$ edges have $r$ binary nodes
$\La\Lt\Lt\Ra$ and $r+1$ leaves.
Substituting $m=n_1=r$, $a=2r$, $c=0$, and $\ell=r+1$, we get
\[
\frac{1}{r}
\binom{r}{r}
\binom{2r}{r+1}
\binom{r-1}{r-1}
=
\frac{1}{r+1}
\binom{2r}{r}
= C_r
\,,
\]
the Catalan numbers, as expected.
Alternatively,
letting $\ell=0$, so the leaves are not specified, we  get again
\[
\frac{1}{r}
\binom{r}{r}
\binom{2r}{0}
\binom{2r}{r-1}
=
\frac{1}{r+1}
\binom{2r}{r}
= C_r
\,.
\]

If we want  to count binary trees in which none of $2r+1$ binary nodes has only one leaf child, then we need exactly
$r$ binary nodes of the form
$\La\Dt\Dt\Ra$ and $r+1$  of the form
$\La\Bx\Bx\Ra$,
giving
($n_1=r$, $n_2=r+1$,
$m=2r+1$, $\ell=a=0$, and $c=2r$) once again
\[
\frac{1}{2r+1}
\binom{2r+1}{r}
=
\frac{1}{r+1}
\binom{2r}{r}
= C_r
\,,
\]
the Catalan number that counts the number of binary trees with $r$ internal nodes of either kind.

Likewise, when all nodes are accounted for, that is, when $v+\ell=n+1$, then $k=u-m=\ell$ in formula (\ref{alt}), which  reduces to
\begin{align}\label{nodes}
\frac{1}{m}
\binom{m}{n_1,\ldots,n_q}
\sum\limits_{i=0}^{e-\ell}
\binom{d+i-1}{i}
\binom{e+b-i-1}{b-1}
\binom{e+a-i}{\ell}
\,,
\end{align}
where $e=m+\ell-a-c-1$.
For example, if a tree has exactly $\ell$ leaves and, hence, $n+1-\ell$ internal (nonleaf) nodes $\La\Lt\Lc\Ra$, then setting $m=a=b=n+1-\ell$ and $c=d=0$, we get (after simplification)
\begin{equation}\label{Narayana}
\frac{1}{n+1}\binom{n+1}{\ell}\binom{n-1}{\ell-1}\,
\end{equation}
since, again, $i$ can only be 0.
This is the Narayana distribution~\cite{Narayana} (\seqnum{A001263}),
which counts $\ell$-leaf ordered trees~\cite{leaf}.

When there is a single composite pattern and $\ell$ additional leaf patterns, we have $m=n_1=q=1$. The main formula (\ref{main}) becomes
\begin{align}\nonumber
&
\sum\limits_{k=1}
\binom{u-k}{\ell}
\sum\limits_{i=0}^{e-k+1}
\binom{d+i-1}{i}
\binom{e+b-i-1}{b+k-2}
\binom{e+a-i}{u-k}
\\{}={} &\nonumber
\sum\limits_{i=0}
\binom{d+i-1}{i}
\sum\limits_{k=1} %^{e-i+1}
\binom{e+b-i-1}{b+k-2}
\binom{e+a-i}{u-k}
\binom{u-k}{\ell}
%\\&\nonumber\sum\limits_{i=0}\binom{d+i-1}{i}\binom{e+a-i}{\ell}\sum\limits_{k=1}^{e-i+1}\binom{e+b-i-1}{b+k-2}\binom{e+a-i-\ell}{n-v+2-k-\ell}
\\{}={} &\nonumber
\sum\limits_{i=0}^{w-a}
\binom{d+i-1}{i}
\binom{w-i}{\ell}
\sum\limits_{k=0}^{w+c-\ell}  %^{w-a-i}
\binom{w-a+b-i-1}{b+k-1}
\binom{w-\ell-i}{w+c-\ell-k}
\,,
\end{align}
where $w=e+a=n-v-c+1$.
Simplifying with Vandermonde's convolution, we derive the following:

\begin{corollary}[Single Pattern]\label{lem:one}
The number of occurrences of a single composite pattern containing $v$ nodes, 
$a$ light triangles, $b$ light ellipses, $c$ dark triangles, and $d$ dark ellipses, along with $\ell$ leaves, among the ordered trees with $n$ edges is
\begin{align}
\sum\limits_{i=0}^{w-a}
\binom{d+i-1}{i}
\binom{w-i}{\ell}
\binom{2w-a+b-\ell-2i-1}{w+b+c-\ell-1}
\,,
\label{one}
\end{align}
where $w=n-v-c+1$.
\end{corollary}

For example, the lone pattern $\La\Dt\Dc\Lc\Lt\Ra$,
with $v=a=b=c=d=1$, occurs twice in the trees with at least three leaves ($\ell=3$) of Figure~\ref{fig:pro} ($n=4$), specifically at the roots of the second and fifth trees.

When there are no leaf patterns ($\ell=0$), this formula becomes
\begin{align}
\sum\limits_{i=0}^{n-v-c-a+1}
\binom{d+i-1}{i}
\binom{2n-2v-a+b-2c-2i+1}{n-v+b}
\,.
\label{leafless}
\end{align}
There are a total of 5  instances of $\La\Dt\Dc\Lc\Lt\Ra$ among all the trees in the figure.

On the other hand,
when there are no dark ellipses in the pattern ($d=0$),
we have instead   
\begin{align}
\binom{n-v-c+1}{\ell}
\binom{2(n-v-c)-a+b-\ell+1}{n-v+b-\ell}
\,.
\nonumber%\label{darkless}
\end{align}

For example,
there are 2 instances of $\La\Dt\Lc\Lt\Ra$ ($a=b=c=v=1$) along with three leaves ($\ell=3$)
in the figure ($n=4$).
In fact, in this situation, where there is only one composite pattern besides the leaves and hence no fear of overlapping patterns, there is no need for dark triangles at all: each $\Dt$ in a pattern may be replaced by $\La\Lt\Lc\Ra$, adding $c$ to each of $a$, $b$, and $v$, and setting $c=0$.

When there are neither dark ellipses nor leaf patterns,
we are left  with just
\begin{align}
\binom{2(n-v-c)-a+b+1}{n-v+b}
\,,
\nonumber%\label{clueless}
\end{align}
which conforms with~\cite[\S 3.2]{Pat} for the $c=0$ case.
Figure~\ref{fig:pro} has $5$ occurrences of $\La\Dt\Lc\Lt\Ra$ all told.

\section{Protected Nodes and Fine Trees}\label{sec:pro}

As  mentioned, an internal (nonleaf) node is deemed {protected} when all its children have offspring~\cite{protected}.
In the size 4 case, as can be seen in Figure~\ref{fig:pro},
there are 6 trees with no protected nodes, 6 with one, and 1 each with two and three.
There are 10 trees with exactly two unprotected nodes, and there
are 4 with only one, including 1 with one of each.
Like sequence~\seqnum{A143362}, but
unlike the enumeration in~\cite{Fine}, roots are  included here in the node count.

Suppose we wish to count the number of trees with $n$ edges, $j\geq 1$ leaves, and no protected nodes at all.
Referring to Figure~\ref{fig:pro}, there are 5 such for $n=4$, $j=3$.
We simply need $m=n+1-j$ patterns $\La\Dc\Bx\Lc\Ra$ for all the unprotected internal nodes, plus $\ell=2j-n-1$ additional leaf patterns $\Bx$ for any remaining leaves.
Letting 
$d=b=m=n+1-j$, $a=c=0$,  $v=2m$, $e=j-1$, and $u=j$ in (\ref{main}),
we get 
\begin{align*}
&
\sum\limits_{k=1}^{j}
\frac{1}{k}
\binom{k}{n-j+1}
\binom{j-k}{2j-n-1}
\sum\limits_{i=0}^{j-1}
\binom{n-j+i}{n-j}
\binom{n-i-1}{k-1}
\binom{j-i-1}{j-k}
\,.
\end{align*}
For the first binomial to be nonzero, we need $n-j+1\leq k$ and for the second, $2j-n-1\leq j-k$.
So $k$ can only take $n-j+1$.
Since all nodes are accounted for, this counts trees.
Summing for all $j$ (starting from $n$ and going down)
and simplifying, we obtain a closed form for sequence \seqnum{A143363}:
\begin{proposition}
The number of ordered trees with $n$ edges and no protected nodes is given by
\begin{align}\label{nopro}
&\sum_{j=0}^{n}
\frac{1}{j+1}
\sum\limits_{i=0}^{n-j-1}
\binom{j+i}{j}
\binom{n-i-1}{j}
\binom{n-j-i-1}{n-2j-1}
\,.
%\\{}={} &\sum_{j=0}^{n}\frac{1}{j+1}\sum\limits_{i=0}^{n-j-1}\binom{j+i}{j}\binom{n-i-1}{2j-i}\binom{2j-i}{j}
\end{align}
\end{proposition}
%\nb{simp}

More generally, suppose we wish to count the number of trees with $n$ edges, $r$ protected nodes, and $s$ unprotected.
For this, we need $r$ patterns $\La\Dt\Dc\Ra$ 
for (nonleaf) protected nodes, 
$s$ patterns $\La\Dc\Bx\Lc\Ra$ 
to account for the unprotected nodes, and $\ell=n+1-r-2s$ pure-leaf patterns.
Letting  
$m=d=r+s$, $a=0$, $b=s$, $c=r$, 
and $e=n-r-s$ in formula (\ref{nodes}), 
and summing for all $j=r+s$,
yields
a closed form for sequence~\seqnum{A143362}:

\begin{proposition}
The number of ordered trees with $n$ edges and $r$ protected nodes is given by
\begin{align}\label{nr}
\sum_{j=r+1}^{n}
\frac{1}{j}
\binom{j}{r}
\sum\limits_{i=j}^{2j-r-1}
\binom{i-1}{j-1}
\binom{n-r-i+j-1}{j-r-1}
\binom{n-i}{n+r-2j+1}
\,.
\end{align}
\end{proposition}

Exchanging the r\^oles of protected and unprotected, we get an analogous enumeration by unprotected nodes:

\begin{proposition}
The number of ordered trees with $n$ edges and $s$ unprotected nodes is given by
\begin{align*}  
\sum_{j=s}^{n}
\frac{1}{j}
\binom{j}{s}
\sum\limits_{i=j}^{s+j-1}
\binom{i-1}{j-1}
\binom{n+s-i-1}{s-1}
\binom{n-i}{n-s-j+1}
\,.
\end{align*}
\end{proposition}

%\section{Fine Forests}

There are 11 protected nodes in the trees displayed in Figure~\ref{fig:pro}, of which
6 are roots and 5 are not.
Call a tree \emph{protected} when its root is.
There are, then, 6 protected 4-edge trees.

The pattern $\La \Dt \Dc \Ra$ matches protected nodes;
hence, occurrences of $\La \Lc \La \Dt \Dc \Ra \Lc \Ra$ correspond to protected nonroots.
The difference between occurrences of those two patterns counts protected roots (a trick used in~\cite[\textsection 3.3]{Pat}).
Plugging in the appropriate values in (\ref{leafless})
($d=c=v=1$, $a=b=0$, $w=n-v-c+1=n-1$ for the first; $d=c=1$, $a=0$, $b=v=2$, $w=n-2$ for the second) 
%and simplifying 
%(by shifts of indices and Vandermonde Convolution) 
gives 
\begin{align}\nonumber
&
\sum\limits_{i=0}
\binom{i}{i}
\binom{2n-2i-3}{n-1}
-
\sum\limits_{i=0}
\binom{i}{i}
\binom{2n-2i-3}{n}
\\{}={}&\label{eq:fine}
\sum\limits_{i=1}
\left[
\binom{2n-2i-1}{n-1}
-
\binom{2n-2i-1}{n}
\right]
%=\sum_{i=1}^{n-1} \frac{i}{n-i}\binom{2n-2i}{n}
\,.
\end{align}
This sum of ballot numbers gives the Fine numbers~\cite{Fine}, listed at \seqnum{A000957}:

\begin{proposition}[{\cite[\S1(F)]{survey}}]\label{thm:fine}
The Fine numbers 
\begin{equation}\label{fine}
\sum_{i=1}^{n-1} \frac{i}{n-i}\binom{2n-2i}{n}
\end{equation}
count protected trees with $n$ edges.
\end{proposition}
\noindent
See also~\cite[\S4]{survey} and references therein.

Let us refer to a leaf on level one, just below the root of a tree, as a \emph{stump}.
Protected trees are stump-free.

Stumps in trees correspond to hills---that is, level 1 peaks---in lattice paths.
The path in Figure~\ref{fig:path} has one hill, just as the tree in Figure~\ref{fig} has one stump.
So, trees sans stumps are one and  the same as paths without hills.
It follows that the above proposition counts hill-less paths, a.k.a.\@ Fine paths.
In other words,
the number of lattice paths of length $2n$ starting and ending on the baseline $y=0$ with no \textit{hills} peaking at $y=1$ is  also counted by the Fine numbers~\cite{Deutsch}.

\section{Tree Stumps and Dyck Hills}\label{sec:stump}

A lone lozenge pattern $\Bx$ in formula (\ref{leaves}) counts leaves.
The total number of leaves in all trees with $n$ edges is known to be $\binom{2n-1}{n-1}=\frac12\binom{2n}n$,
which is half the nodes in all the trees~\cite{DY},~\cite[Cor.\@ 2.1]{Enum}.
The number of stumps in a set of trees is
the number of leaves
minus the number of non-stump leaves, the
latter counted by $\La\Lc\La\Lc\Bx\Lc\Ra\Lc\Ra$.
There are 35 leaves in Figure~\ref{fig:pro}, 14 of which are stumps.

Generalizing this example and that of the prior section,
 we focus on the  case of a single unambiguous tree pattern $p$.
Let $p'$ be the embedded pattern $\La\Lc \,p\, \Lc\Ra$.
We can use (\ref{leafless}) to count occurrences of both $p$ and $p'$, the disparity
between them only being that the latter has one more node $\La\cdots\Ra$ and two extra ellipses $\Lc$.
Taking the difference between their two enumerations, we have
\begin{align}
\label{diff}
\sum\limits_{i=0}^{w-a}
\binom{d+i-1}{i}
\left[
\binom{2w-a+b-2i-1}{w+b+c-1}
-
\binom{2w-a+b-2i-1}{w+b+c}
\right]\,,
\end{align}
where the parameters refer to $p$ and $w=n-v-c+1$.
This brings us to the following:

\begin{theorem}[Root Pattern]\label{lem:root}
The number of $n$-edge trees with an occurrence of an unambiguous pattern at the root---%
a pattern containing $v$ nodes, 
$a$ light triangles, $b$ light ellipses, $c$ dark triangles, and $d$ dark ellipses---%
is
\begin{align}
\label{rooted}
&
\sum\limits_{i=0}^{w+(b-a)/2-1}
\frac{a+b+2c+2i}{2w-a+b-2i}
\binom{d+i-1}{i}
\binom{2w-a+b-2i}{w+b+c}
\,,
\end{align}
where $w=n-v-c+1$.
\end{theorem}

As a trivial example, trees with a unary root---called \emph{planted} trees---%
are characterized by the root pattern $\La\Lt\Ra$,
with $v=a=1$ and $b=c=d=0$, forcing $i=0$:
\[
\frac{1}{2n-1}
\binom{2n-1}{n}
 = C_{n-1}
\,,
\]
just as one should expect.
As long as $n>1$, the planted trees are also characterized by
$\La\Dt\Ra$, with $v=c=1$ and $a=b=d=0$:
\[
\frac{1}{n-1}
\binom{2n-2}{n}
= C_{n-1}
\,.
\]
If the pattern is $\La\Lc\Ra$, then all (nonleaf)  trees are counted,
with $v=b=1$ and $a=c=d=0$:
\[
\frac{1}{2n+1}
\binom{2n+1}{n}
= C_n
\,.
\]
On the other hand, the root pattern is $\La\Dc\Ra$ counts
protected trees having no stumps
with $v=d=1$ and $a=b=c=0$:
\[
\sum\limits_{i=0}^{n-1}
\frac{i}{n-i}
\binom{2n-2i}{n}
\,,
\]
the Fine numbers, as before,
except that $\La\Dc\Ra$ may match an edgeless ($n=0$) tree $\La\Ra$,
while the pattern $\La \Dt \Dc \Ra$, which we used earlier, cannot.

Suppose now that we wish to count 3-protected trees, that is, trees whose root has no childless children or grandchildren~\cite{kpro}.
The pattern 
\[
\big\langle\underbrace{\La\Dt\Dc\Ra\cdots\La\Dt\Dc\Ra}_{c\mbox{ \scriptsize times}}\big\rangle
\]
counts nodes with $c$ children, all of whom are (2-) protected.
Using formula (\ref{diff}) for root patterns, with $d=c$, $v=c+1$, and $a=b=0$,
and summing for $c$, and  assuming $n\geq  3$, we obtain 
\begin{align*}
&
\sum_{c=1}^{n/2}
\sum\limits_{i=0}^{n-2c}
\binom{c+i-1}{i}
\left[
\binom{2n-4c-2i-1}{n-c-1}
-
\binom{2n-4c-2i-1}{n-c}
\right]
\end{align*}
for the number of 3-protected trees with $n$ edges.
This should be equivalent to the closed form
\begin{align*}
\sum_{c=1}^{(2n-1)/5}  %{n/2}
(-1)^{c-1}
\left[
%\binom{2n-5c-5}{n-c-6}-\binom{2n-5c-5}{n-c-3}
\binom{2n-5c-1}{n-c-2}-\binom{2n-5c-1}{n-c+1}
\right]
\,,
\end{align*}
which is the $k=3$ case of the general formula for $k$-protection
given in \cite[Prop.\@ 5.5]{prono}.
For $n=4$, 
this is $[2-0]-[0-0]=2$, the last two  trees  in Figure~\ref{fig:pro}.
It likewise counts 
the number of hill-free Dyck paths  having also no peaks at level  2;
see sequence \seqnum{A114627}.

Trees with $n$ edges and 
root pattern
\[\big\langle \overbrace{ \Lt\Lc \cdots \Lt\Lc  \Lt\Lc}^{r \mbox{ \scriptsize times}}  \big\rangle\]
 are equinumerous to ordered forests (the order of the trees in the forest matters) with $n$ edges and $r$ trees---disallowing edgeless trees in the forest.
 We need the triangles because an ellipsis alone can match an empty sequence of subtrees; this way we have partitioned the children of the root into $r$ groups, each representing a nonempty tree in the forest.
 There may be many ways to partition, but each corresponds to a different forest.
We derive
\begin{align}\nonumber
&
\sum\limits_{k=0}
\binom{n-1}{r+k-1}
\binom{n}{n-k}
-
\sum\limits_{k=0}
\binom{n}{r+k+1}
\binom{n-1}{n-k-1}
\\{}={}&\nonumber
\sum\limits_{k=0}
\binom{n-1}{n-r-k}
\binom{n}{k}
-
\sum\limits_{k=0}
\binom{n}{n-r-1-k}
\binom{n-1}{k}
\\{}={}&
\binom{2n-1}{n-r}
-
\binom{2n-1}{n-r-1}
\,.\label{for}
\end{align}
Accordingly,
\begin{proposition}
The number of ordered forests consisting of $r$ nonleaf trees and containing $n$ edges altogether 
is
\begin{align*}%\label{fort}
\frac{r}{n}
\binom{2n}{n-r}
\,.
\end{align*}
\end{proposition}
\noindent
This is Catalan's triangle, as per~\cite{Shapiro}, listed at~\seqnum{A039598},
with special cases: 
$r=2$~\seqnump{A002057};
$r=3$~\seqnump{A003517};
$r=4$~\seqnump{A003518};
$r=5$~\seqnump{A003519}; and
$r=n-2$~\seqnump{A014106}.

In \cite{1700}, it was shown that the sequence \seqnum{A001700} enumerates $n$-edge forests, and indeed it is easy to see that
\[\sum_{r=1}^n (\ref{for}) = \binom{2n-1}n\,.\]

The pattern $\La\Dc\Bx\Lc\Ra$ counts eldest leaves.
Since any one tree can have only one eldest stump,
subtracting occurrences of  $\La\Lc\!\La\Dc\Bx\Lc\Ra\!\Lc\Ra$ 
counts trees with at least one stump---making it unprotected.
There are 8 such unprotected trees in Figure~\ref{fig:pro}.
Using formula (\ref{diff}),  with $v=2$, $a=c=0$, and $b=d=1$, 
we obtain the following enumeration of unprotected trees:
\begin{align}
&
\sum\limits_{i=0}
\binom{2n-2i-2}{n-1}
-
\sum\limits_{i=0}
\binom{2n-2i-2}{n}
\label{untree}
\end{align}
for $n>1$.
This is the difference between two recorded sequences, 
\seqnum{A014300} and \seqnum{A172025},
and provides an alternative closed form for the latter.

As must be the case, the protected (\ref{eq:fine}) and unprotected (\ref{untree}) trees add up to all trees, which are counted by the Catalan numbers:
\begin{align*}
&
\sum\limits_{i=0}
\left[
\binom{2n-2i-1}{n}
 + 
\binom{2n-2i-2}{n}
\right]
-
\sum\limits_{i=0}
\left[
\binom{2n-2i-1}{n-1}
+
\binom{2n-2i-2}{n-1}
\right]
\\{}={}&
\binom{2n}{n}
-
\binom{2n}{n-1}
%= \frac{1}{n+1}\binom{2n}{n}
=C_n
\;.
\end{align*}

When there can be more than one occurrence of a pattern at the root,
formula (\ref{rooted}) still counts root occurrences, if not trees.
For example, the pattern $\La\Lc\Bx\Lc\Ra$ occurs once at the root for each stump.
Letting $v=2$, $a=c=d=0$, and $b=2$ in (\ref{rooted}), we obtain
\[
\frac{1}{n+1}
\binom{2n+2}{n} 
\,.
\]

It follows that
\begin{proposition}
The number of stumps (level-one leaves) within ordered trees is counted by the Catalan numbers.
\end{proposition}
\noindent
This comes as no surprise, since a stump (or for that matter any fixed subtree of the root) splits ordered trees with $n$ edges into two trees with a total of $n-1$ edges (or $n$ minus the size of any other divider), and the Catalan numbers are well-known to satisfy the corresponding recurrence,
as illuminated by the binary trees.

Consider now what a tree with  $r\geq 1$ stumps looks like.
The pattern 
\[\big\langle  \Dc \!\overbrace{\Bx \Dc \cdots \Dc \Bx \Dc}^{r \mbox{ \scriptsize times}}  \big\rangle\]
with $r$ embedded leaf patterns matches a node with exactly $r$
childless children.
Subtracting  occurrences of  $\La \Lc \La  \Dc \Bx \Dc \cdots \Dc \Bx \Dc \Ra  \Lc \Ra$ gives the number of roots
with $r$ stumps.
With the appropriate values in (\ref{diff}), viz.\@
$a=b=c=0$ and $v=d=r+1$, we obtain the following enumeration:

\begin{proposition}\label{thm:stump}
The number of trees with $n$ edges and $r$
stumps is
 $1$ for $r=n$,
is $0$ for $r=n-1$, and, for $r<n-1$, is
\begin{align}
%\nonumber&\sum_{i=0}\binom{r+i}{r}\left[\binom{2n-2r-2i-1}{n-r-1}-\binom{2n-2r-2i-1}{n-r}\right]\\{}={}&
\sum_{i}
\frac{i}{n-r-i}
\binom{r+i}{r}
\binom{2n-2r-2i}{n-r}
\,.
\label{hill}
\end{align}
\end{proposition}
\noindent
For $n=4$, as in Figure~\ref{fig},
the number of trees with $r=$ 0, 1, 2, 3, 4 stumps is
6, 4, 3, 0, 1, respectively.

When $r=0$, we retrieve the Fine numbers (\ref{fine}).
When $r=n-2$, this yields $n-1$, 
with all but one of the $n-1$ children of the root being childless.
When $r=n-3$, this yields $2(n-2)$, 
with two ways to bless each of $n-2$ children with two progeny.

On account of the bijection with lattice paths,
Proposition~\ref{thm:stump} also enumerates paths of length $2n$ sporting $r>0$ hills.

\section{Discussion}
Our tree enumeration formul\ae---with their ability to capture a wide variety of patterns---%
have been used here, in particular, to count trees with various conditions on the number of leaves just below the root.
We note that many of the summations we have seen have double the index appearing within the binomial coefficients 
%(eqs.~\ref{one},\ref{leafless},\ref{nopro},\ref{nr},\ref{fine},\ref{rooted},\ref{untree},\ref{hill}).
(eqs.~\ref{one}--\ref{hill}).

\subsection*{Acknowledgement}
I thank a referee for suggesting applying the method herein to the enumeration of $(k>2)$-protected trees.

% ----------------------------------------------------------------
\bibliography{More}

\begin{thebibliography}{10}

\bibitem{Fine}
David Callan.
\newblock Some bijections and identities for the {Catalan} and {Fine} numbers.
\newblock {\em S\'eminaire Lotharingien de Combinatoire}, 53:Art. B53e, 2006.
\newblock \url{http://www.mat.univie.ac.at/~slc/s/s53callan.ps}.

\bibitem{protected}
Gi-Sang Cheon and Louis~W. Shapiro.
\newblock Protected points in ordered trees.
\newblock {\em Applied Mathematics Letters}, 21(5):516--520, 2008.
\newblock \url{http://dx.doi.org/10.1016/j.aml.2007.07.001}.

\bibitem{kpro}
Keith Copenhaver.
\newblock $k$-{Protected} vertices in unlabeled rooted plane trees.
\newblock {\em Graphs and Combinatorics}, 33(2):347--355, March 2017.
\newblock \url{https://arxiv.org/pdf/1606.00083.pdf}.

\bibitem{DY}
Balakrishnan Dasarathy and Cheng Yang.
\newblock A transformation on ordered trees.
\newblock {\em Computer Journal}, 23(2):161--164, 1980.

\bibitem{1700}
Nachum Dershowitz.
\newblock 1700 forests, 2016.
\newblock arXiv:1608.08740 [cs.LO].

\bibitem{Enum}
Nachum Dershowitz and Shmuel Zaks.
\newblock Enumerations of ordered trees.
\newblock {\em Discrete Mathematics}, 31(1):9--28, 1980.

\bibitem{Pat}
Nachum Dershowitz and Shmuel Zaks.
\newblock Patterns in trees.
\newblock {\em Discrete Applied Mathematics}, 25:241--255, 1989.
\newblock \url{http://nachum.org/papers/PatternsInTrees.pdf}.

\bibitem{CL}
Nachum Dershowitz and Shmuel Zaks.
\newblock The {Cycle Lemma} and some applications.
\newblock {\em European J. of Combinatorics}, 11:35--40, 1990.
\newblock \url{http://nachum.org/papers/CL.pdf}.

\bibitem{More}
Nachum Dershowitz and Shmuel Zaks.
\newblock More patterns in trees: Up and down, young and old, odd and even.
\newblock {\em SIAM J. on Discrete Mathematics}, 23(1):447--465, 2009.
\newblock \url{http://nachum.org/papers/UpDown.pdf}.

\bibitem{Deutsch}
Emeric Deutsch.
\newblock Dyck path enumeration.
\newblock {\em Discrete Math.}, 204:167--202, 1999.

\bibitem{survey}
Emeric Deutsch and Louis Shapiro.
\newblock A survey of the {Fine} numbers.
\newblock {\em Discrete Math.}, 241(1--3):241--265, 2001.

\bibitem{DM}
Aryeh Dvoretzky and Theodore Motzkin.
\newblock A problem of arrangements.
\newblock {\em Duke Mathematical Journal}, 14:305--313, 1947.

\bibitem{prono}
Clemens Heuberger and Helmut Prodinger.
\newblock Protection number in plane trees.
\newblock {\em Appl. Anal. Discrete Math.}, 11:314--326, 2017.
\newblock \url{http://pefmath.etf.rs/vol11num2/AADM-Vol11-No2-314-326.pdf}.

\bibitem{Narayana}
Tadepalli~V. Narayana.
\newblock A partial order and its applications to probability.
\newblock {\em Sankhya}, 21:91--98, 1959.

\bibitem{leaf}
John Riordan.
\newblock Enumeration of plane trees by branches and endpoints.
\newblock {\em J. Combin. Theory Ser. A}, 19:215--222, 1975.

\bibitem{Shapiro}
Louis~W. Shapiro.
\newblock A {Catalan} triangle.
\newblock {\em Discrete Mathematics}, 14(1):83--90, 1976.

\bibitem{OEIS}
Neil~J. Sloane.
\newblock The on-line encyclopedia of integer sequences, 1964--.
\newblock \url{http://oeis.org}.

\bibitem{Stanley}
Richard~P. Stanley.
\newblock {\em Catalan Numbers}.
\newblock Cambridge University Press, 2015.

\end{thebibliography}
\bibliographystyle{plain}

\end{document}